\documentclass[journal]{IEEEtran}

\usepackage{amsmath,amsthm,amssymb,amscd,mathrsfs}
\usepackage{graphics}
\usepackage{ulem}       
\usepackage{url}
\usepackage{tikz}
\usepackage{pgf}
\usepackage{amsfonts, mathdots}
\usepackage{color}
\usepackage{subfig}
\usepackage[ruled,vlined]{algorithm2e}
\usepackage{latexsym}
\usepackage{cite}
\usepackage{multirow}

\DeclareMathOperator*{\argminA}{arg\,min}

\def\N{\mathbb N}

\def\R{\mathbb R}

\def\P{\mathbb P}

\newcommand{\diag}{\operatorname{diag}\,}
\def\tr{{\rm tr\,}}

\def\Pn{\mathbb{P}_n}
\def\Hn{\mathbb{H}_n}

\numberwithin{equation}{section}
\newtheorem{theorem}{Theorem}[section]
\newtheorem{lemma}[theorem]{Lemma}

\theoremstyle{definition}

\ifCLASSINFOpdf
\else
\fi

\begin{document}
%
\title{Barycenter Estimation of Positive Semi-Definite Matrices with Bures-Wasserstein Distance}
%
%
%

\author{Jingyi~Zheng,~\IEEEmembership{}
        Huajun~Huang,~\IEEEmembership{}
        Yuyan~Yi,~\IEEEmembership{}
        Yuexin~Li,~\IEEEmembership{}
        and~Shu-Chin Lin,~\IEEEmembership{}
\thanks{J. Zheng, H. Huang, Y. Yi, and Y. Li are with the Department of Mathematics and Statistics, Auburn University, Auburn, AL, 36849, USA. S. Lin is with Center for Neuropsychiatric Research, National Health Research Institutes, Taiwan. correspondence e-mail: jingyi.zheng@auburn.edu.}}


%
%

\markboth{Journal of \LaTeX\ Class Files,~Vol.~14, No.~8, August~2015}
{Zheng \MakeLowercase{\textit{et al.}}: Bures-Wasserstein distance}
%


\maketitle

\begin{abstract}
Brain-computer interface (BCI) builds a bridge between human brain and external devices by recording brain signals and translating them into commands for devices to perform the user's imagined action. The core of the BCI system is the classifier that labels the input signals as the user's imagined action. The classifiers that directly classify covariance matrices using Riemannian geometry are widely used not only in BCI domain but also in a variety of fields including neuroscience, remote sensing, biomedical imaging, etc. However, the existing Affine-Invariant Riemannian-based methods treat covariance matrices as positive definite while they are indeed positive semi-definite especially for high dimensional data. Besides, the Affine-Invariant Riemannian-based barycenter estimation algorithms become time consuming, not robust, and have convergence issues when the dimension and number of covariance matrices become large. To address these challenges, in this paper, we establish the mathematical foundation for Bures-Wasserstein distance and propose new algorithms to estimate the barycenter of positive semi-definite matrices efficiently and robustly. Both theoretical and computational aspects of Bures-Wasserstein distance and barycenter estimation algorithms are discussed. With extensive simulations, we comprehensively investigate the accuracy, efficiency, and robustness of the barycenter estimation algorithms coupled with Bures-Wasserstein distance. The results show that Bures-Wasserstein based barycenter estimation algorithms are more efficient and robust.

\end{abstract}


%
\IEEEpeerreviewmaketitle

\section{Introduction}
%
%
%
%


\IEEEPARstart{B}{rain} computer interface (BCI) builds a bridge between human brain and external devices by translating the brain signals into instructions for the external devices to perform the user's imagined actions. The core of BCI system is the classifier that classify the brain signals into one of the commands for the external devices. The brain signals can be captured by multi-channel Electroencephalography (EEG) \cite{schirrmeister2017deep, zheng2021time}, functional magnetic resonance imaging (fMRI) \cite{qiu2015manifold, varoquaux2010detection}, and other neuroimaging techniques, which leads to BCI data being mostly spatial and temporal. To capture the spatial and temporal pattern of BCI data, covariance matrices are widely used, and the BCI classifiers are trained by directly classifying covariance matrices using Riemannian geometry \cite{barachant2010riemannian, barachant2013classification, miah2019eeg}.
Besides BCI, covariance matrices are also widely used to capture the structure of complex data in various fields such as
computer vision \cite{chen2020covariance, porikli2006covariance, sivalingam2010tensor},
natural language processing \cite{jagarlamudi2011improving, zhang2014discriminatively},
domain adaption \cite{cui2014flowing, zhang2018aligning},
remote sensing \cite{he2018remote, eklundh1993comparative},
biomedical imaging \cite{qiu2015manifold, varoquaux2010detection},
and many others \cite{yang2013solar, meyer2009factor, huang2015face}.

Previous works have shown that covariance matrices are often treated as positive definite matrices, and analyzed on the manifold of positive definite matrices, denoted as \(\Pn\), using Affine-Invariant Riemannian metric, which has been a popular choice of used Riemannian metric due to its many great mathematical properties. Considering \(\mathbb{P}_n\) as an open sub-manifold of \({\mathbb{M}}_n\), at any point $A\in \Pn$ (i.e., any positive definite matrix $A$), the Affine-Invariant Riemannian  metric  on the tangent space $T_A\Pn$ is defined as
\begin{equation}\label{innerR}
    \langle X,Y \rangle_A= \tr (A^{-1} X A^{-1} Y).
\end{equation}
The corresponding distance function for two positive definite matrices $A$ and $B$, named Affine-Invariant Riemannian  (AI) distance, is
 \begin{equation}\label{disR}
    d_{AI}(A,B)=\big{(}\sum_{i=1}^{n}\log^2\lambda_i(A^{-1}B)\big{)}^{1/2},
\end{equation}
where $\lambda_i(A^{-1}B), i=1,\dots,n$ is the eigenvalues of $A^{-1}B$.
With respect to the metric (\ref{innerR}), any two points on $\Pn$ (i.e., any two positive definite matrices) is joined by the following geodesic:
\begin{equation}\label{JeodR}
    \gamma_{AI}(t)=A^{1/2}(A^{-1/2}BA^{-1/2})^t A^{1/2}, \ \  t\in [0,1].
\end{equation}
The geometric mean of the two matrices, denoted as $A\#B$, is the mid point of the geodesic:
\begin{equation}\label{geom}
    \gamma_{AI}(1/2)=A^{1/2}(A^{-1/2}BA^{-1/2})^{1/2} A^{1/2}.
\end{equation}
It can be easily shown that $d_{AI}(A,\gamma_{AI}(t))=td_{AI}(A,B)$. Therefore, we have $d_{AI}(A,A\#B)=\frac{1}{2}d_{AI}(A,B)$.

When analyzing a set of matrices, one of the most important measure is the barycenter of matrices. Like the arithmetic mean of a set of numbers, the barycenter is a measure of the central tendency of a set of matrices on the manifold. It is also used when calculating the standard deviation and variance of a set of matrices. Most importantly, it is the foundation for developing classification models for matrices. Using AI distance, statistical and computational methods have been developed \cite{barachant2010riemannian, barachant2013classification, miah2019eeg, arsigny2007geometric, jayasumana2013kernel, huang2015log, lin2019riemannian, pennec2020manifold} to analyze covariance matrices including the barycenter estimation and classification of matrices.

However, computing AI distance involves matrix inverse and eigenvalue decomposition as shown in (\ref{disR}), which are time-consuming and computationally unstable especially for large matrices. Furthermore, instead of being strictly positive definite, a covariance matrix is in fact a positive semi-definite (PSD) matrix, especially for high dimensional data. As the dimension of covariance matrix increases, the likelihood of having zero eigenvalues grows substantially, which invalidates the use of AI distance. Therefore, we propose to analyze covariance matrices directly on the manifold of PSD matrices coupled with Bures-Wasserstein distance.

In this paper, we first establish the mathematical foundation for the Bures-Wasserstein (BW) distance by proving its properties and studying the retractions maps of the manifold of PSD matrices. To estimate the central tendency of a set of matrices, we then propose three algorithms to estimate the Fr\'echet mean (i.e,  barycenter) of matrices using BW distance. Extensive simulations are conducted to comprehensively investigate the efficiency and robustness of BW distance, as well as the accuracy, efficiency, and robustness of the proposed barycenter estimation algorithms. The remaining of the paper is organized as follows.
In Section~\ref{methods} (A), we establish the mathematical properties of BW distance and the retraction maps of the manifold. In Section~\ref{methods} (B), we propose three algorithms for estimating the barycenter of matrices on the manifold. In Section~\ref{result}, we present the simulation results and discuss the efficiency and robustness of BW distance and the proposed algorithms for barycenter estimation, and a full comparison with the widely used AI distance. In Section~\ref{conclusion}, we summarize our contributions and conclude the paper.

\section{Methodology} \label{methods}

By viewing $\Pn$ as the quotient manifold \({\mathbb{P}}_n=\frac{\mathrm{GL(n)}}{\mathrm{U(n)}}\), where \(\mathrm{GL(n)}\) is the set of nonsingular matrices, \cite{bhatia2018on} proposed the Bures-Wasserstein (BW) distance
\begin{eqnarray} \notag
d_{BW}(A,B) &=& \left[ \tr(A+B)-2\tr(A^{1/2}BA^{1/2})^{1/2} \right]^{1/2} \\
&=& \left[ \tr(A+B)-2\tr(AB)^{1/2} \right]^{1/2} \notag
\end{eqnarray}
where $A,B \in \Pn$.
Notice that the above BW distance $d_{BW}(A,B)$ coincides with the Wasserstein distance between two Gaussian distribution with the same mean and covariance matrices being $A$ and $B$, respectively.

The geodesic from $A$ to $B$ is defined as $\gamma(t): [0,1]\to \Pn$
\begin{equation}
\gamma(t)=(1-t)^2A+t^2 B+t(1-t)[(AB)^{1/2}+(BA)^{1/2}] \notag
\end{equation}

Indeed, the above BW distance and the geodesic can be extended to
the set $\overline{\P}_n$ of $n\times n$ positive semi-definite (PSD) matrices by viewing \(\overline{\P}_n=\frac{\mathrm{M(n)}}{\mathrm{U(n)}}\)
since \(\overline{\mathrm{GL(n)}}=\mathrm{M(n)}\). Therefore, for any two PSD matrices $A$ and $B$, the BW distance is
\begin{equation} \label{bw-dist}
d_{BW}(A,B) = \left[ \tr(A+B)-2\tr(AB)^{1/2} \right]^{1/2}
\end{equation}
where $A,B \in \overline{\P}_n$. And the geodesic from $A$ to $B$, denoted as $A\diamond_t B$, is
\begin{equation} \label{gamma}
    A\diamond_t B:=\gamma(t)=(1-t)^2A+t^2 B+t(1-t)[(AB)^{1/2}+(BA)^{1/2}]
\end{equation}
Same as $A\#B$, $A\diamond_{1/2} B$ is called Wasserstein mean and denoted as $A \diamond B$
\begin{equation}\label{Wmean}
    A\diamond_{1/2} B=\frac{1}{4}[A+B+(AB)^{1/2}+(BA)^{1/2}].
\end{equation}

\subsection{Establish the mathematical foundation for BW distance}

Few studies have discussed the mathematical properties of BW distance \cite{bhatia2018on, bhatia2019inequalities, THANWERDAS2023163, hwang2022two, thanwerdas2022riemannian, kim2020inequalities, hwang2019bounds,  massart2020quotient, malago2018wassersteingd}, especially on $\overline{\P}_n$.
In the following, we study some mathematical properties of BW distance and the retraction maps of the manifold $\overline{\P}_n$. These properties  provide nice insights and intuitions about BW metric and BW mean.

Given $X\in \mathrm{M}_n$, let $|X|=(X^*X)^{1/2}$ denote the PSD part in the polar decomposition of $X$.

\begin{theorem}
For $A, B\in\overline{\P}_n$ and $t\in\R$,
\begin{equation}\label{BW geodesic form}
A\diamond_t B =B\diamond_{1-t}A = |(1-t)A^{1/2}+t U^* B^{1/2}|^2.
\end{equation}
in which $U$ is a certain unitary matrix occurring in a polar decomposition of $B^{1/2}A^{1/2}$:
\begin{equation}\label{polar BA half}
B^{1/2}A^{1/2}=U|B^{1/2}A^{1/2}|=
U(A^{1/2}BA^{1/2})^{1/2},
\end{equation}
or equivalently, $A^{1/2}B^{1/2}=U^*|A^{1/2}B^{1/2}|$.  Moreover,
when  $(1-t)A+t|B^{1/2}A^{1/2}|$ is PSD (e.g. when $t\in [0,1]$),
\begin{equation}\label{BW prop 1}
|(A\diamond_t B)^{1/2}A^{1/2}|=
(1-t)A+t|B^{1/2}A^{1/2}|.
\end{equation}
\end{theorem}

\begin{proof}
If $A^{1/2}$ is nonsingular, i.e. $A\in\P_n$, then for any unitary matrix $U$ in the polar decomposition \eqref{polar BA half}, we have
\begin{subequations}
\begin{eqnarray}
    A^{1/2}U^* B^{1/2} &=& A^{1/2} (U^* B^{1/2}A^{1/2})A^{-1/2}
    \notag\\
    &=& A^{1/2} (A^{1/2}BA^{1/2})^{1/2} A^{-1/2}
   \notag \\ &=& (A^{1/2} A^{1/2}BA^{1/2}A^{-1/2})^{1/2}
    \notag \\ &=& (AB)^{1/2}.
\end{eqnarray}
If $A^{1/2}$ is singular, we may find a sequence of nonsingular positive definite matrices $\{A_i\}_{i=1}^{\infty}\subseteq \P_n$ such that
$\lim_{i\to\infty}A_i=A$. Let $B^{1/2}A_i^{1/2}=U_i|B^{1/2}A_i^{1/2}|$ be the polar decomposition for $i=1,2,3,\ldots$
Since the unitary group of degree $n$  is compact, the sequence $\{U_i\}_{i=1}^{\infty}$
has a convergent subsequence $\{U_{i_t}\}_{t=1}^{\infty}$. Let  $U:=\lim_{t\to\infty} U_{i_t}$.
Then
\begin{eqnarray}
\notag
A^{1/2}U^*B^{1/2}
    &=&\lim_{t\to\infty} A_{i_t}^{1/2}U_{i_t}^*B^{1/2}
\\ &=& \lim_{t\to\infty} (A_{i_t}B)^{1/2}=(AB)^{1/2}.
\qquad
\end{eqnarray}
\end{subequations}
In both cases,
\begin{equation}
    B^{1/2}UA^{1/2} = (A^{1/2}U^* B^{1/2})^*
= (BA)^{1/2}.
\end{equation}
So the geodesic from $A$ to $B$ can be expressed as:
\begin{eqnarray}
 A\diamond_t B
 &=& B\diamond_{1-t}A
\notag\\
&=& (1-t)^2A+t^2 B+t(1-t)[(AB)^{1/2}+(BA)^{1/2}]
\notag\\
&=& |(1-t)A^{1/2}+t U^* B^{1/2}|^2.
\notag\end{eqnarray}
Moreover,
\begin{eqnarray}
 |(A\diamond_t B)^{1/2}A^{1/2}|
&= & |\, |(1-t)A^{1/2}+t U^* B^{1/2}|A^{1/2}|
\notag \\
&=& |\, [(1-t)A^{1/2}+t U^* B^{1/2}]A^{1/2}|
\notag \\
&=& |(1-t)A+t|B^{1/2}A^{1/2}|\, |.
\label{BW prop 2}
\end{eqnarray}
When  $(1-t)A+t|B^{1/2}A^{1/2}|$ is PSD, we have
\begin{equation}  \notag
|(A\diamond_t B)^{1/2}A^{1/2}|
=(1-t)A+t|B^{1/2}A^{1/2}|.
\end{equation}
Therefore, the theorem is proved.
\qedhere
\end{proof}

We remark that \eqref{BW prop 2} still holds even when the Hermitian matrix $(1-t)A+t|B^{1/2}A^{1/2}|$ is not PSD.
When both $A$ and $B$ are singular, some unitary matrix $U$ that satisfies \eqref{polar BA half}
may not satisfy \eqref{BW geodesic form}.

\begin{theorem} Let $A, B\in \overline{\P}_n$.
\begin{enumerate}
\item
If $r,t\in\R$ and $(1-t)A+t|B^{1/2}A^{1/2}|\in \overline{\P}_n$,
then
\begin{equation} \label{BW prop 3}
A\diamond_r(A\diamond_t B) = A\diamond_{rt}B.
\end{equation}
\item If  $r,s,t\in\R$ satisfy that $(1-x)A+x|B^{1/2}A^{1/2}|\in \overline{\P}_n$ for $x\in\{s,t\}$,
then
\begin{equation} \label{BW prop 4}
(A\diamond_s B)\diamond_r (A\diamond_t B)=A\diamond_{(1-r)s+rt}B.
\end{equation}
\end{enumerate}
\end{theorem}

\begin{proof}
\eqref{BW prop 3} is a special case of \eqref{BW prop 4} by choosing $s=0$. We will
show that \eqref{BW prop 3} also implies \eqref{BW prop 4} after proving  \eqref{BW prop 3}.

Suppose $r,t\in\R$ such that $(1-t)A+t|B^{1/2}A^{1/2}|\in \overline{\P}_n$. Assume that $A\in\P_n$ first (the case of singular $A$ can be done by continuous extension). By \eqref{BW prop 2} and \eqref{BW prop 1},
\begin{eqnarray}
&&|[A\diamond_r(A\diamond_t B)]^{1/2}A^{1/2}|
 \notag\\
&=& |(1-r)A+r|(A\diamond_t B)^{1/2}A^{1/2}|\, |
 \notag\\
&=& |(1-r)A+r[(1-t)A+t|B^{1/2}A^{1/2}|\, ]\, |
 \notag\\
&=& |(1-rt)A+rt|B^{1/2}A^{1/2}|\, |
 \notag\\
&=&
|(A\diamond_{rt} B)^{1/2}A^{1/2}|.
\end{eqnarray}
Hence there is a unitary matrix $V$ such that
$$[A\diamond_r(A\diamond_t B)]^{1/2}A^{1/2}=V(A\diamond_{rt} B)^{1/2}A^{1/2}.$$
By assumption, $A$ is nonsingular, so that
 $[A\diamond_r(A\diamond_t B)]^{1/2}=V(A\diamond_{rt} B)^{1/2}$. We get
\eqref{BW prop 3}.

Now suppose $r,s,t\in\R$ satisfy that $(1-x)A+x|B^{1/2}A^{1/2}|\in \overline{\P}_n$ for $x\in\{s,t\}$.
Again we assume that $A\in\P_n$ first and then extend the result continuously to singular $A$. Suppose $t\le s$ without loss of generality. Then by \eqref{BW prop 3},
\begin{eqnarray} \notag
&&(A\diamond_s B)\diamond_r (A\diamond_t B)
 \notag\\
&=& (A\diamond_s B)\diamond_r[A\diamond_{\frac{t}{s}} (A\diamond_s B)]
 \notag\\
&=& (A\diamond_s B) \diamond_r[ (A\diamond_s B) \diamond_{1-\frac{t}{s}} A]
 \notag\\
&=& (A\diamond_s B) \diamond_{\frac{rs-rt}{s}} A
 \notag\\
&=&  A\diamond_{1-\frac{rs-rt}{s}} (A\diamond_s B)
 \notag\\
 \notag&=& A\diamond_{(1-r)s+rt} B.
\end{eqnarray}
We get \eqref{BW prop 4}.
\end{proof}

Let $U$ be  the unitary matrix in the  polar decomposition of
$B^{1/2}A^{1/2}$ as in \eqref{polar BA half}.
Bhatia, Jain, and Lim showed that \cite[Theorem 1]{bhatia2018on}:
\begin{equation}
d_{BW}(A,B)=\|A^{1/2}-B^{1/2}U\|_{F}
\end{equation}
It implies the following property about BW distance.

\begin{theorem}
For $A, B\in\overline{\P}_n$ and $t\in\R$, if $(1-t)A+t|B^{1/2}A^{1/2}|$ is PSD (e.g. when $t\in [0,1]$), then
\begin{equation}\label{BW prop 5}
d_{BW}(A,A\diamond_t B)
=|t|d_{BW}(A,B).
\end{equation}
\end{theorem}

\begin{proof} We prove for the nonsingular $A$ case. The singular $A$ case can be done by continuous extension.
Let $U$  and $V$ be the unitary matrix in the  polar decomposition of
$B^{1/2}A^{1/2}$ and $(A\diamond_t B)^{1/2}A^{1/2}$, respectively:
\begin{eqnarray}
B^{1/2}A^{1/2}
&=& U|B^{1/2}A^{1/2}|,
\notag \\
(A\diamond_t B)^{1/2}A^{1/2}
&=& V|(A\diamond_t B)^{1/2}A^{1/2}|
\notag \\
&=& V[(1-t)A+t|B^{1/2}A^{1/2}|\, ].
\notag
\end{eqnarray}
Taking Hermitian transposes on the above equalities, we get
\begin{eqnarray}
A^{1/2}B^{1/2}
&=& |B^{1/2}A^{1/2}| U^*,
\notag \\
A^{1/2}(A\diamond_t B)^{1/2}
&=& [(1-t)A+t|B^{1/2}A^{1/2}|\, ]V^*.
\notag \end{eqnarray}
Therefore,
\begin{eqnarray}
(A\diamond_t B)^{1/2}V
&=&A^{-1/2}[(1-t)A+t|B^{1/2}A^{1/2}|\, ]
\notag \\
&=& (1-t)A^{1/2}+tB^{1/2}U.
\notag \end{eqnarray}
By \cite[Theorem 1]{bhatia2018on},
\begin{eqnarray}  \notag
d_{BW}(A,A\diamond_t B)
&=&
 \|A^{1/2}-(A\diamond_t B)^{1/2}V\|_{F}
 \\ \notag
&=&\|tA^{1/2}-t B^{1/2}U\|_{F}
\\
&=& |t|d_{BW}(A,B).
\notag
\end{eqnarray}
We get \eqref{BW prop 5}.
\end{proof}

As a Riemannian manifold, points in $\overline{\P}_n$ (i.e., PSD matrices) can be projected to the flat tangent space as illustrated in Figure \ref{fig:manifold}. The logarithm function maps points on the manifold to the tangent space while exponential function maps the points on the tangent space back to the manifold. With the log and exp functions, matrices on the manifold can be easily project to the flat tangent space where methods that work in Euclidean space are applicable. Therefore, derivation of the log and exp functions is extremely important and fundamental for the further analysis.

Let $\Hn$ (resp. $\mathrm{U(n)}$) denote the set of $n\times n$ real symmetric (resp. orthogonal) matrices.

\begin{figure}[!tbp]
\centering
  \includegraphics[width=0.48\textwidth]{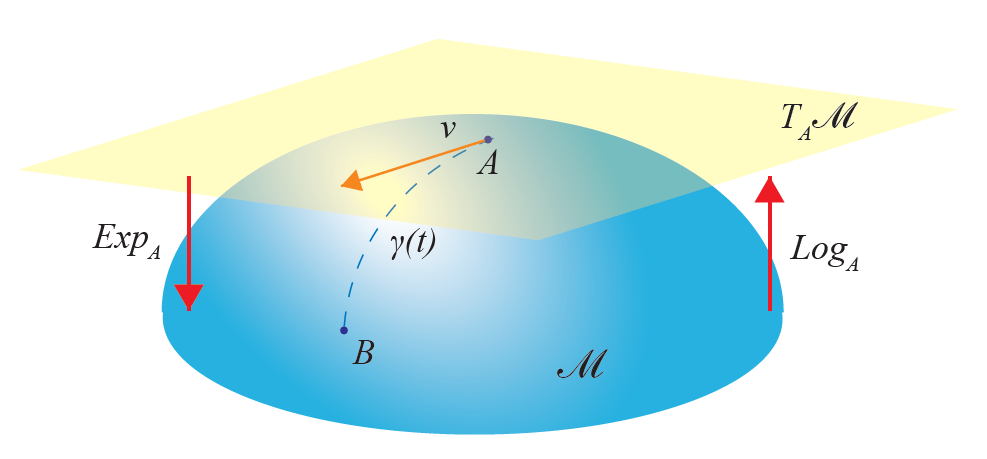}
  \caption{Manifold and its tangent space at point $A \in \overline{\P}_n$}
  \label{fig:manifold}
\end{figure}

\begin{lemma}\label{BWBasicProperties}
  For  $A, B\in \overline{\Pn}$, $X\in \Hn$ such that $\exp_A X$ is well-defined, $U\in\mathrm{U(n)}$, and $t\in\R$,
\begin{eqnarray}\label{BW unitary eq 1}
    d_{BW}(U^*AU, U^*BU) &=& d_{BW}(A,B),
    \\\label{BW unitary eq 2}
    (U^*AU)\diamond_t (U^*BU) &=& U^*(A\diamond_t B)U,
    \\\label{BW unitary eq 3}
    \log_{U^*AU}(U^* BU) &=& U^* (\log_{A}{B}) U,
    \\\label{BW unitary eq 4}
    \exp_{U^*AU}(U^*XU) &=& U^* (\exp_{A}X)U.
\end{eqnarray}
\end{lemma}

\begin{proof}
The proofs of \eqref{BW unitary eq 1} and \eqref{BW unitary eq 2} are straightforward by \eqref{bw-dist} and  \eqref{gamma}. Taking $\left.\frac{d}{dt}\right|_{t=0}$ on both sides of \eqref{BW unitary eq 2},
we get \eqref{BW unitary eq 3}. Then \eqref{BW unitary eq 4} follows.
\end{proof}

According to the spectral decomposition, every $A\in\overline{\P}_n$ can be written as
$A=U\Lambda U^*$ for a nonnegative diagonal matrix $\Lambda$ and a unitary matrix $U\in \mathrm{U(n)} $.
Lemma \ref{BWBasicProperties} implies that we can transform the BW metric around $A$ to
that around the diagonal matrix $\Lambda$ and simplify the computations.

The log function on $\Pn$ under BW metric has been described in \cite{bhatia2018on}, \cite{massart2020quotient} and \cite{thanwerdas2022riemannian}. We add an approximation of $\log_{A}(A+tX)$ when $tX$ is nearby $0$ as follows.


\begin{theorem} \label{theoremlog}
For any $A, B\in \Pn$, $X\in \Hn$,  and  $t \in \R$ sufficiently close to $0$, we have
\begin{eqnarray}
\log_{A}(B) &=& (AB)^{1/2}+(BA)^{1/2}-2A,\qquad \label{log_A B}\\
\log_{A}(A+tX) &=& t X+\O (t^{2}). \label{themlog}
\end{eqnarray}
\end{theorem}

\begin{proof}
By \eqref{gamma}, the tangent vector of geodesic
$\gamma(t)$  at $A$  is
\begin{equation}\notag
\log_A(B)= \gamma'(0) = (AB)^{1/2}+(BA)^{1/2}-2A.
\end{equation}
For $X\in \Hn$ and  $t \in \R$ sufficiently close to $0$, we have $A+tX\in\Pn$, so that
\begin{eqnarray}  \notag
\log_{A}(A+tX)=[A(A+tX)]^{1/2}+[(A+tX)A]^{1/2}-2A.
\end{eqnarray}
 Moreover,
\begin{eqnarray}  \notag
[A(A+tX)]^{1/2}\notag
&=& [A^{1/2}(A^2+t A^{1/2}X A^{1/2})A^{-1/2}]^{1/2} \\ \notag
&=& A^{1/2}(A^2+t A^{1/2}X A^{1/2})^{1/2} A^{-1/2},
\\
\notag
[(A+tX)A]^{1/2}
&=& [A^{-1/2}(A^2+t A^{1/2}X A^{1/2})A^{1/2}]^{1/2} \\ \notag
&=& A^{-1/2}(A^2+t A^{1/2}X A^{1/2})^{1/2} A^{1/2}
\end{eqnarray}
Let $C:=(A^2+t A^{1/2}X A^{1/2})^{1/2}$. Then
\begin{equation} \label{log expansion}
    \log_{A}(A+tX)= A^{1/2}CA^{-1/2}+A^{-1/2}CA^{1/2}-2A.
\end{equation}
When $t$ is close to 0, $C$ can be expressed as a power series of $t$:
\begin{eqnarray}  \label{log: C}
C
= (A^{2}+tA^{1/2}XA^{1/2})^{1/2}
= A+tZ+\O (t^{2}).
\end{eqnarray}
Taking squares, we have
\begin{eqnarray}  \notag
A^2+ t A^{1/2}XA^{1/2}
 &=& [A+tZ+\O (t^{2})]^2
\\
&=& A^2+t(AZ+ZA)+\O (t^{2}).
\qquad\quad
\end{eqnarray}
Therefore, $AZ+ZA=A^{1/2}XA^{1/2}$. By \eqref{log expansion} and \eqref{log: C},
\begin{eqnarray} \notag
&&\log_{A}(A+tX)  \\ \notag
&=& A^{\frac{1}{2}}(A+tZ )A^{-\frac{1}{2}}+A^{-\frac{1}{2}}(A+tZ )A^{\frac{1}{2}}
-2A+\O(t^2) \qquad
\\ \notag
&=& t A^{-\frac{1}{2}}(AZ+ZA)A^{-\frac{1}{2}}+\O(t^2) \\
&=& tX+\O(t^2).
\qquad
\notag \qedhere
\end{eqnarray}
\end{proof}

The exponential map has been studied in \cite{malago2018wassersteingd} and \cite{THANWERDAS2023163}. Here we give a concise form of  the exponential map and  provide its exact domain.
We also  provide an approximation of $\exp_{A}(X)$ when $X$ is a Hermitian matrix nearby  $0$.
Let $A\circ B$ denote the Hadamard product of matrices $A$ and $B$ of the same size.

\begin{lemma} \label{lemmaexp}
Let $A=\diag(\lambda_{1},..., \lambda_{n})\in \Pn$ be a positive diagonal matrix.
Denote $W:=\left(\frac{1}{\lambda_{i}+\lambda_{j}}\right)_{n \times n}$. Then for every Hermitian matrix $X\in \Hn$ such that
$I_n+W\circ X$ is PSD,
\begin{eqnarray}  \label{expdiag}
\begin{split}
\exp_{A}(X) & =A+X+ (W \circ X )A(W \circ X ).
\end{split}
\end{eqnarray}
 In particular, for $t \in \R$ sufficiently close to $0$,
\begin{eqnarray}  \label{lemma2}
\exp_{A}(tX)=A+tX+\O(t^{2}).
 \end{eqnarray}
\end{lemma}

\begin{proof}
Firstly, let $B=\exp_A(X)$ and
\begin{equation}\label{exp: C}
C:=(A^{1/2}BA^{1/2})^{1/2}.
\end{equation} Then \eqref{log_A B} and the assumption $A=\diag(\lambda_{1},\ldots,\lambda_{n})$  imply that
\begin{eqnarray}\notag
2A+X & =& (AB)^{1/2}+(BA)^{1/2}  \\ \notag
     & =& A^{1/2}(A^{1/2}BA^{1/2})^{1/2}A^{-1/2}
     \\\notag &&\quad
    +A^{-1/2}(A^{1/2}BA^{1/2})^{1/2}A^{1/2}
    \\ \notag & = & A^{1/2}CA^{-1/2}+A^{-1/2}CA^{1/2}
    \\ \notag & = & \left(\frac{\lambda_i+\lambda_j}{\lambda_i^{1/2} \lambda_j^{1/2}} \right)_{n\times n}\circ C.
\\
\notag
 C &=& \left(\frac{\lambda_i^{1/2} \lambda_j^{1/2}}{\lambda_i+\lambda_j} \right)_{n\times n}\circ (2A+X)
 \\
 &=& A+\left(\frac{\lambda_i^{1/2} \lambda_j^{1/2}}{\lambda_i+\lambda_j} \right)_{n\times n}\circ X.
\end{eqnarray}
By \eqref{exp: C}, $C$ must be PSD, which is equivalent to that
$A^{-1/2}CA^{-1/2}=I_n+W\circ X$ is PSD. In such a case,
\begin{eqnarray*}
B &=& \exp_{A}(X)=A^{-1/2}C^{2}A^{-1/2}
\\ &=&  A^{-1/2} \left(A+\left(\frac{\lambda_i^{1/2} \lambda_j^{1/2}}{\lambda_i+\lambda_j} \right)_{n\times n}\circ X \right)^{2} A^{-1/2}
\\ &=& A+X+ (W\circ X)A(W\circ X).
\end{eqnarray*}
The expression of \eqref{expdiag} is obtained.
Replace $X$ by $tX$ for small $t$, then \eqref{lemma2} follows.
\end{proof}

Similarly, the exponential function and approximation that are derived in \textbf{Lemma \ref{lemmaexp}}  for the case of $A$ being a positive diagonal matrix  also can be further defined for general cases.

\begin{theorem}
Suppose $A \in \Pn$ has the spectral decomposition $A=U\Lambda U^{*}$, where
$U$ is a unitary matrix and  $\Lambda=\diag(\lambda_1,...,\lambda_n)$.
Denote $W=\left(\frac{1}{\lambda_{i}+\lambda_{j}} \right)_{n \times n}$.
Then for every Hermitian matrix $X \in \Hn$ such that $I_n+W\circ X_{U}$ is PSD where
$X_{U} :=U^*XU$, we have
\begin{equation}\label{expgen}
 \exp_{A}(X)=A+X+U\left[ (W\circ X_{U})\Lambda (W\circ X_{U}) \right]U^*.
\end{equation}
In particular, when  $t\in\R$ is sufficiently closed to zero, we have
\begin{eqnarray} \label{thmexp}
\exp_{A}(tX)=A+tX+\O (t^{2}).
\end{eqnarray}
\end{theorem}

\begin{proof}
By Lemma \ref{lemmaexp}, under the assumption of $X$,
\begin{eqnarray*}
\exp_{A}X &=& \exp_{U\Lambda U^*}X=U\exp_{\Lambda} (U^*XU) U^*
\\ &=& U[\Lambda+X_U+(W\circ X_U)\Lambda(W\circ X_U)]U^*
\\ &=& A+X+U[W\circ X_U)\Lambda(W\circ X_U]U^*.
\end{eqnarray*}
Replace $X$ by $tX$, then we get \eqref{thmexp}.
\end{proof}

\subsection{Estimate the Barycenter of PSD matrices with BW distance}

To summarize a set of numbers, arithmetic mean is used to measure the central tendency while standard deviation is used to measure the variation or dispersion of the numbers. Similarly, we also need metrics to characterize the distribution of a set of PSD matrices on $\overline{\P}_n$. With BW distance, we estimate the central tendency of a set of PSD matrices $A_1,\ldots, A_m \in \overline{P}_n$ by the Fr\'echet mean, defined as the following:
\begin{equation} \label{Fmean}
    \overline{A}(A_1,\ldots, A_m)=\argminA_{X \in \overline{\P}_n}
    \sum_{i=1}^m{d_{BW}^2(A_i, X)}
\end{equation}
The Fr\'echet mean is also called the barycenter in literatures (e.g. \cite{bhatia2018on, yger2017riemannian}). With the Fr\'echet mean, we can further quantify the dispersion of matrices around their Fr\'echet mean via the Fr\'echet variance as the following:
\begin{equation} \label{sd}
    \sigma^2=\frac{1}{m}\sum_{i=1}^m{d_{BW}^2(A_i, \overline{A})}
\end{equation}

To estimate the Fr\'echet mean for a set of matrices on $\overline{P}_n$, we propose three methods: Inductive Mean Algorithm, Projection Mean Algorithm, and Cheap Mean Algorithm. The error tolerance in the stopping criteria of each algorithm is denoted as $\varepsilon$. Details are discussed in the following.

\begin{figure}[!tbp]
  \centering
    \includegraphics[width=0.45\textwidth]{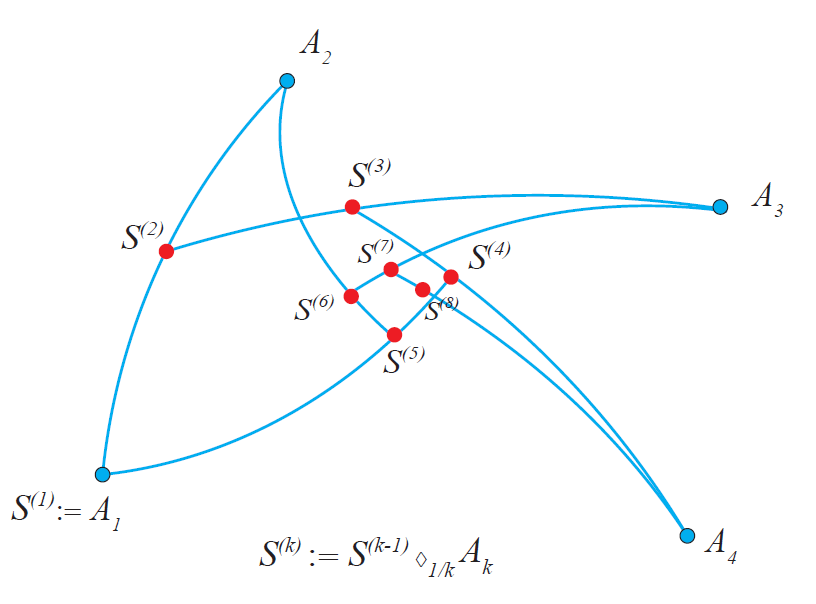}
    \caption{Illustration of the Inductive Mean Algorithm.}
    \label{fig:IMalgorithm}
\end{figure}

{\bf Inductive Mean Algorithm}
Given $m$ PSD matrices $A_1,\ldots, A_m \in \overline{P}_n$, the {\bf Inductive Mean Algorithm} estimates their Fr\'echet mean using the geodesic that connects two points on the manifold. Details of the algorithm is summarized as the following, and we illustrate the process using four matrices in Figure \ref{fig:IMalgorithm}.

\begin{enumerate}
\item Define the sequence $\{A_{k}\}_{k\in\N}$ such that $A_{k}=A_{k+m}=A_{k+2m}=\cdots$ for all $k\in\N$.

\item Let $S^{(1)}:=A_1.$
For $k=2, 3,\ldots,$ let
\begin{equation*}
\begin{split}
S^{(k)} &:= S^{(k-1)}\diamond_{\frac{1}{k}} A_{k} \\
&=  \frac{(k-1)^2}{k^2}  S^{(k-1)}+\frac{1}{k^2} A_{k}+ \\
& \qquad \frac{k-1}{k^2}\left [(S^{(k-1)}A_{k})^{1/2}+(A_{k}S^{(k-1)})^{1/2}\right ].
\end{split}
\end{equation*}

\item The limit of $\{S^{(k)}\}_{k\in\N}$ is the Fr\'echet mean of $\{A_1,\ldots,A_m\}$ with $d_{BW}$:
\begin{equation}
\lim_{k\to\infty} S^{(k)} = \overline{A}(A_1,\ldots, A_m).
\end{equation}

In practice, the iteration process stops when $d_{BW}(S^{(k)},S^{(k-1)}) \le \varepsilon$, where $S^{(k)}$ is the estimated Fr\'echet mean.

\end{enumerate}

By the design of this algorithm, points in $\{S^{(k)}\}_{k\in\N}$ are located in the compact region bounded by the geodesics connecting $A_1,\ldots,A_m$, and the convergence point of $\{S^{(k)}\}_{k\in\N}$ is the unique convergence point. Therefore, the Inductive Mean algorithm is valid for estimating the barycenter of PSD matrices.

The computation process of the algorithm is simple, and is not sensitive to the size and number of matrices as long as the maximal distance between two matrices in $A_1,\ldots,A_m$ is bounded. The convergence rate is uniform but slow compared to the following two algorithms.

{\bf Projection Mean Algorithm}
Different from the Inductive Mean Algorithm which uses the geodesic on the manifold,  {\bf Projection Mean Algorithm} on $(\overline{P}_n, d_{BW})$ leverages the Log and Exp  functions that project matrices between the manifold and tangent space. The iteration process of the algorithm is summarized in the following:

\begin{enumerate}
\item Let $S^{(0)}:=\frac{1}{m}\sum_{j=1}^m A_j$.

\item Suppose $S^{(\ell)}$ is known for some $\ell\in\N$, update $S^{(\ell+1)}$ as follows:

(a) Project $\{A_1,\ldots,A_m\}$ onto the tangent space at $S^{(\ell)}$ by \eqref{log_A B}:
        \begin{equation} \notag
            X_j^{(\ell)}:=(S^{(\ell)}A_j)^{1/2}+(A_jS^{(\ell)})^{1/2}-2S^{(\ell)}.
        \end{equation}
(b) Find the arithmetic mean of the projection vectors.
    \begin{equation*}
        \begin{split}
            X^{(\ell)} & :=\frac{1}{m}\sum_{j=1}^{m} X_{j}^{(\ell)} \\
            & = \frac{1}{m}\sum_{j=1}^{m}\left[(S^{(\ell)}A_j)^{1/2}+(A_jS^{(\ell)})^{1/2}\right]-2S^{(l)}.
        \end{split}
        \end{equation*}
(c) $S^{(\ell+1)}$ is updated as the exponential of $X^{(\ell)}$ at $S^{(\ell)}$ by \eqref{expgen}.
     \begin{equation*}
           S^{(\ell+1)} := S^{(\ell)} +  X^{(\ell)} +
             U^{(\ell)}[(W^{(\ell)}\circ X_U^{\ell})\Lambda^{(\ell)} (W^{(\ell)}\circ X_U^{\ell})]{U^{(\ell)}}^*.
     \end{equation*}
    where
\begin{eqnarray*}
            S^{(\ell)} &=& U^{(\ell)} \Lambda^{(\ell)} {U^{(\ell)}}^*,
            \quad U^{(\ell)}\in \mathrm{U(n)},\\
            \Lambda^{(\ell)} &=& \diag(\lambda_{1}^{(\ell)}, \lambda_{2}^{(\ell)}, \cdots, \lambda_{n}^{(\ell)})\\
        W^{(\ell)} &:=& \begin{pmatrix}
             \frac{1}{\lambda_{i}^{(\ell)}+\lambda_{j}^{(\ell)}}
           \end{pmatrix}_{n\times n}\\
        X_U^{(\ell)} &=& U^{(\ell)*}X^{(\ell)}U^{(\ell)}.
\end{eqnarray*}

\item The limit of $\{S^{(l)}\}_{l\in\N}$ is the approximation of the Fr\'echet mean of $\{A_1,\ldots,A_m\}$ with $d_{BW}$:
\begin{equation}
\lim_{l\to\infty} S^{(l)} = \overline{A}(A_1,\ldots, A_m).
\end{equation}
In practice, the iteration process stops when $d_{BW}(S^{(l)},S^{(l-1)}) \le \varepsilon$, where $S^{(l)}$ is the estimated Fr\'echet mean.

\end{enumerate}

Compared to the Inductive Mean algorithms, the Projection Mean algorithm converges much faster and has much less computational cost. For instance, when $\varepsilon = 10^{-3}$, the Projection Mean algorithm usually converges within $10$ iterations regardless of the number and dimension of matrices, while the Inductive Mean algorithm takes up to $\O(10^3)$ iterations.

{\bf Cheap Mean Algorithm}
Similar to the Projection Mean Algorithm, the {\bf Cheap Mean Algorithm} also leverages the Log and Exp functions to project matrices between the manifold and the tangent space.
Differently, the Cheap Mean algorithm also updates the original matrices when updating the barycenter. The details of the iterative process are summarized in the following:


\begin{enumerate}
    \item Let $A_k^{(0)}=A_k$ for $k=1,\ldots,m.$

    \item Suppose $A_1^{(\ell)},\ldots,A_m^{(\ell)}$ are known for some $\ell\in\N$. For each $k\in\{1,\ldots,m\}$, we project the geodesic curves connecting $A_k^{(\ell)}$ to $A_1^{(\ell)},\ldots,A_m^{(\ell)}$ onto the tangent space at $A_k^{(\ell)}$. Then find the arithmetic mean of the projection vectors. The exponential of this arithmetic mean at $A_k^{(\ell)}$ is denoted by $A_k^{(\ell+1)}$.

    By \eqref{log_A B}, the projection of   $A_j^{(\ell)}$ onto the tangent space at $A_k^{(\ell)}$ is
        \begin{equation} \notag
            X_{kj}^{(\ell)}:=(A_k^{(\ell)}A_j^{(\ell)})^{1/2}+(A_j^{(\ell)}A_k^{(\ell)})^{1/2}-2A_k^{(\ell)}.
        \end{equation}
     The arithmetic mean of the projection vectors is
        \begin{equation*}
        \begin{split}
            X_{k}^{(\ell)} & :=\frac{1}{m}\sum_{j=1}^{m} X_{kj}^{(\ell)} \\
            & = \frac{1}{m}\sum_{j=1}^{m}\left[(A_k^{(\ell)}A_j^{(\ell)})^{1/2}+  (A_j^{(\ell)}A_k^{(\ell)})^{1/2}\right]
             -2A_k^{(\ell)}. \qquad
        \end{split}
        \end{equation*}
        The spectral decomposition of $A_k^{(\ell)}$ is
        \begin{equation*}
        \begin{split}
            A_k^{(\ell)}
            &=U_k^{(\ell)} \Lambda_k^{(\ell)} {U_k^{(\ell)}}^*,
            \quad U_k^{(\ell)}\in \mathrm{U(n)},\\
            \Lambda_k^{(\ell)}
            &=\diag(\lambda_{k1}^{(\ell)}, \lambda_{k2}^{(\ell)}, \cdots, \lambda_{kn}^{(\ell)}).
        \end{split}
        \end{equation*}
        Denote the matrices
        \begin{eqnarray*}
        W_k^{(\ell)} &:=& \begin{pmatrix}
             \frac{1}{\lambda_{ki}^{(\ell)}+\lambda_{kj}^{(\ell)}}
           \end{pmatrix}_{n\times n},
        \\
        X_{k,U}^{(\ell)} &:=& {U _{k}^{(\ell)}}^*X_{k}^{(\ell)} U _{k}^{(\ell)}.
        \end{eqnarray*}
        By \eqref{expgen}, the exponential of $X_{k}^{(\ell)}$ at $A_{k}^{(\ell)}$ is
        \begin{equation*}
        \begin{split}
            A_k^{(\ell+1)} &:= A_k^{(\ell)}+X_{k}^{(\ell)}
            \\
            &+U_{k}^{(\ell)}\left[(W_{k}^{(\ell)}\circ X_{k,U}^{(\ell)})\Lambda_{k}^{(\ell)}(W_{k}^{(\ell)}\circ X_{k,U}^{(\ell)}) \right] {U _{k}^{(\ell)}}^*
            \end{split}
        \end{equation*}

    \item All sequences $\{A_k^{(\ell)}\}_{\ell\in\N}$ for $k=1,\ldots,m$ converge to the same limit,
    which is called the Cheap Mean $\overline{A}'(A_1,\ldots, A_m)$:
    \begin{equation*}
    \begin{split}
        \overline{A}'(A_1,\ldots, A_m) & := \lim_{\ell\to\infty} A_1^{(\ell)} = \lim_{\ell\to\infty} A_2^{(\ell)} \\
        & =\cdots =  \lim_{\ell\to\infty} A_m^{(\ell)}.
    \end{split}
    \end{equation*}
    In practice, the iteration process stops when $d_{BW}(\frac{1}{m}\sum_{k=1}^{m}A_{k}^{(l)},\frac{1}{m}\sum_{k=1}^{m}A_{k}^{(l-1)}) \le \varepsilon$, where $\frac{1}{m}\sum_{k=1}^{m}A_{k}^{(l)}$ is the estimated Fr\'echet mean.

\end{enumerate}

When $m=2$, the Cheap Mean Algorithm produces the true Fr\'echet Mean. When $m>2$, the Cheap Mean $\overline{A}'(A_1,\ldots, A_m)$ is an approximation of the true Fr\'echet Mean \cite{bini2011anote}.
In practice, the computation cost of this algorithm is relatively low and the convergence rate is competitive to the Projection Mean Algorithm. The algorithm is much faster compared to the inductive mean algorithm, however slower compared with the Projection algorithm.

\section{Experimental Results} \label{result}

Extensive simulations are conducted to comprehensively investigate the robustness of BW distance, the accuracy, efficiency, and robustness of estimating the Fr\'echet Mean of PSD matrices using BW distance on the manifold $\overline{P}_n$. Both efficiency and robustness are compared with the commonly used AI distance (\ref{innerR}). Considering the fact that AI distance only works for positive definite matrices, small values ($O(10^{-3})$) are added to the zero-eigenvalues to make PSD matrices positive definite in our simulations. Table \ref{tab:simulation setting} summarizes the setting of the parameters used in the simulations.

\begin{table}[htbp]
\normalsize
    \caption{Simulation Parameter Setting} \label{tab:simulation setting}
    \centering
    \begin{tabular}{p{5cm}p{3cm}}
        \hline
        Simulation Parameter & Values \\
        \hline
        $n$: Dimension of Matrices & $[5,10,20,30,50,100]$ \\
        $m$: Number of Matrices & $[5,10,20,30,50,100]$ \\
        $p$: Proportion of close-to-zero eigenvalues & $[0.1, 0.2, 0.4, 0.6, 0.8]$ \\
        \hline
    \end{tabular}
\end{table}

\begin{figure}[htbp]
\centering
  \includegraphics[width=0.48\textwidth]{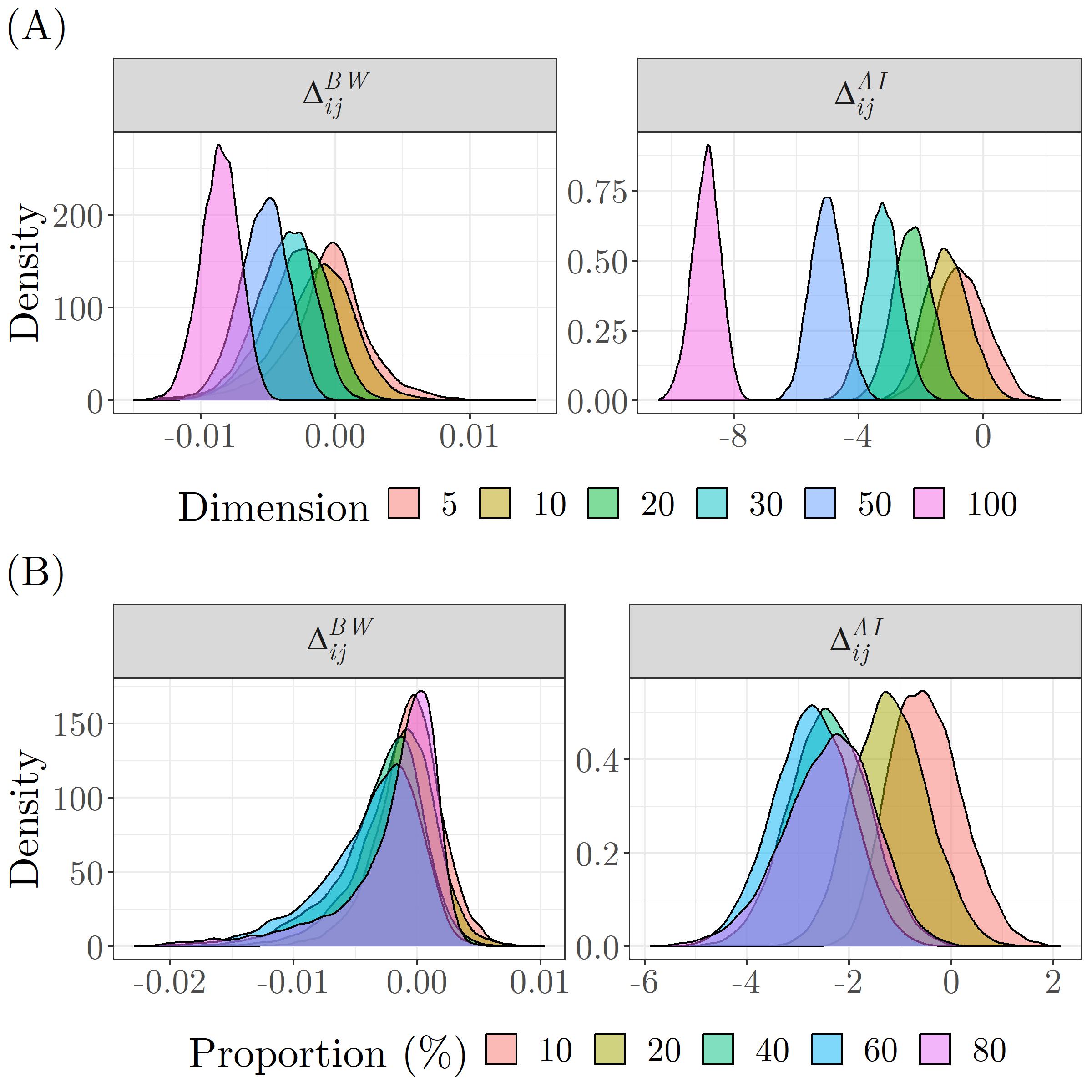}
  \caption{Empirical distributions of $\Delta_{ij}^{AI}$ and $\Delta_{ij}^{BW}$ with (A) varying matrix dimension (B) varying proportion of close-to-zero eigenvalues. }
  \label{fig:dist compare}
\end{figure}

\subsection{Robustness of BW distance}

The robustness of a distance measure is quantified by the changes in distance when matrices are contaminated by small perturbations. The smaller the changes are, the more robust the distance measure is. In this simulation, we test the robustness of BW distance, compare it with AI distance, and also investigate the factors that influence the robustness of BW distance.

Specifically, $100$ pairs of positive definite matrices are randomly generated with $n$ being 10 and $p$ being $0.2$, denoted as $(A_i, B_i), i = 1, \ldots, 100$. Then, $100$ pairs of Hermitian perturbation matrices, denoted as $(E_{ij}^A, E_{ij}^B), j = 1, \ldots, 100$, are randomly generated for each pair of $(A_i, B_i)$, with their spectral norms on the same scale as the smallest eigenvalues of $(A_i, B_i)$, i.e., $10^{-3}$.
The contaminated matrices $(\tilde{A}_{ij}, \tilde{B}_{ij})$ are obtained by adding the perturbation matrices on $(A_i, B_i)$, i.e., $\tilde{A}_{ij} = A_i + E_{ij}^A, \tilde{B}_{ij} = B_i + E_{ij}^B$.
The robustness of distance measures are quantified as the difference in the distances between the two matrices with and without perturbation:

$$\Delta_{ij}^{AI} = d_{AI}(A_{i}, B_{i}) - d_{AI}(\tilde{A}_{ij}, \tilde{B}_{ij})$$
$$\Delta_{ij}^{BW} = d_{BW}(A_{i}, B_{i}) - d_{BW}(\tilde{A}_{ij}, \tilde{B}_{ij})$$
where $i,j = 1, \ldots, 100$.

With $p$ being $0.2$, the empirical distributions of $\Delta_{ij}^{AI}$ and $\Delta_{ij}^{BW}$ with varying matrix dimension $n$ are shown in Figure \ref{fig:dist compare}(A). Overall, $\Delta^{BW}$ is in much smaller scale than $\Delta^{AI}$, which implies that BW distance is much more robust than AI distance when matrices are affected by perturbations. Besides, the matrix dimension also affect the robustness of both BW and AI distances (all p-values $< 0.01$ for Kruskal-Wallis and Dunn test). The higher the matrix dimension is, the less robust the distance is.

With $n$ being $10$, the robustness of distances is also investigated with respect to the proportion of close-to-zero eigenvalues $p$. Figure \ref{fig:dist compare}(B) shows the empirical distribution of $\Delta_{ij}^{AI}$ and $\Delta_{ij}^{BW}$ with varying $p$. Overall, the scale of $\Delta^{BW}$ is much smaller than $\Delta^{AI}$, i.e., BW distance is more robust than AI distance. However, $p$ has less impact on the robustness of both distances than $n$ does (p-values $< 0.01$ for Kruskal-Wallis test, but not all p-values are significant for Dunn tests with $0.01$ significance level, especially for AI distance).

In summary, BW distance is more robust when matrices are affected by perturbations. Besides, BW distance requires much less computing time than AI distance does, which is not hard to show by comparing equation (\ref{disR}) and (\ref{bw-dist}).

\subsection{Robustness of Barycenter for Two Matrices}

The Fr\'echet Mean of two matrices is simply the mid point of the geodesic, i.e., $A\#B$ (\ref{geom}) with AI distance and $A\diamond_{1/2} B$ (\ref{Wmean}) with BW distance. The robustness of barycenter refers to whether the barycenter would be affected when the two matrices are contaminated by small perturbations. In this simulation, we investigate the robustness of the two barycenters $A\diamond_{1/2} B$ and $A\#B$ and the contributing factors of the robustness.

$200$ pairs of positive definite matrices are randomly generated with $n$ being 10 and $p$ being $0.2$, denoted as $(A_i, B_i), i = 1, \ldots, 200$. For each pair of $(A_i, B_i)$, $100$ Hermitian perturbation matrices $(E_{ij}^A, E_{ij}^B), j = 1, \ldots, 100$ are also randomly generated with their spectral norms being $O(10^{-3})$. Then the Fr\'echet Means of $(A_i, B_i)$ are denoted as $M_{i}^{AI}$ and $M_{i}^{BW}$ for AI and BW distance respectively. For matrices with perturbations, the Fr\'echet Means are noted as $\tilde{M}_{ij}^{AI}$ and $\tilde{M}_{ij}^{BW}$ as follows:
$$ M_{i}^{AI} = A_i \# B_i, \quad  M_{i}^{BW} = A_i \diamond_{1/2} B_i$$
$$ \tilde{M}_{ij}^{AI} = \tilde{A}_{ij} \# \tilde{B}_{ij},  \quad \tilde{M}_{ij}^{BW} = \tilde{A}_{ij} \diamond_{1/2} \tilde{B}_{ij}$$
where $\tilde{A}_{ij} = A_i + E_{ij}^A, \tilde{B}_{ij} = B_i + E_{ij}^B, i = 1, \ldots, 200, j = 1, \ldots, 100$ are matrices with random perturbations.

\begin{figure}[htbp]
\centering
  \includegraphics[width=0.48\textwidth]{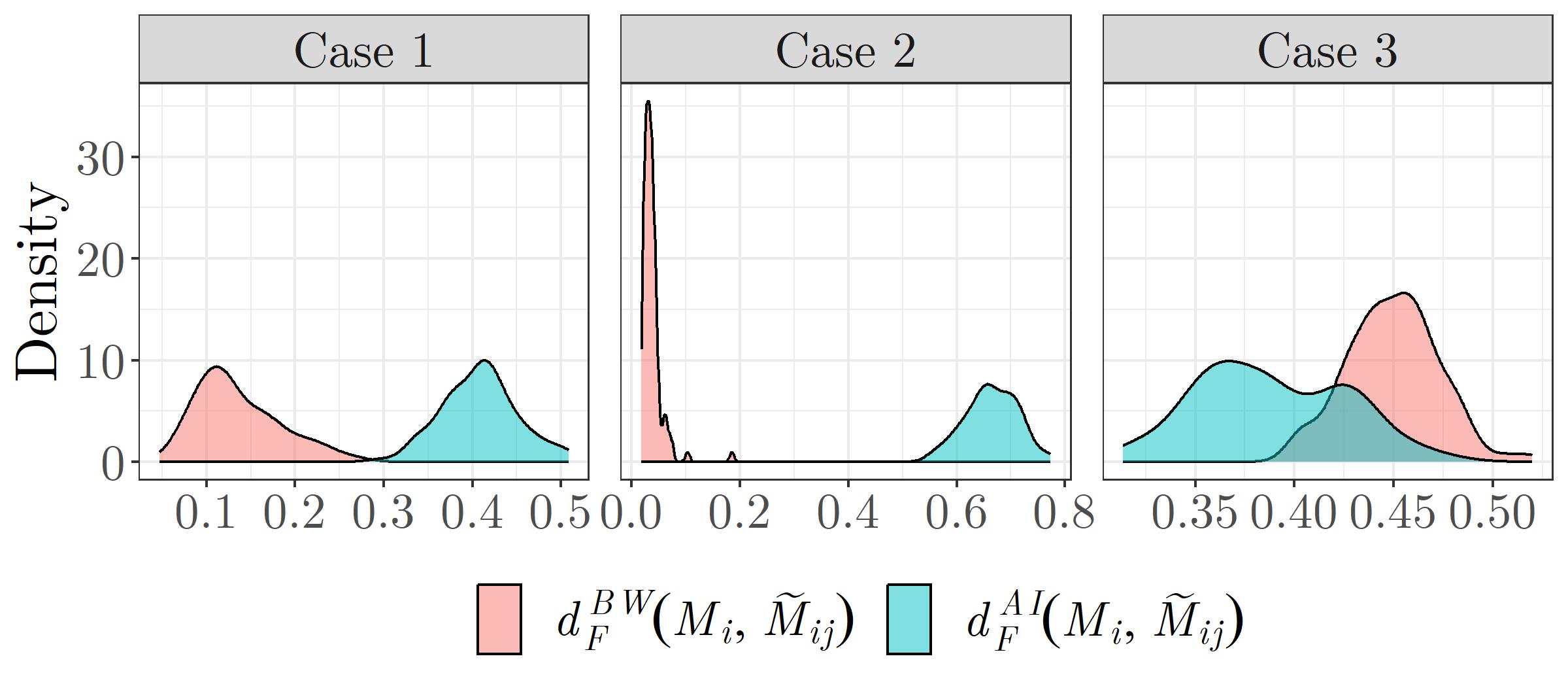}
  \caption{Empirical distributions of $d_{F}^{BW}(M_{i},\tilde{M}_{ij})$ and $d_{F}^{AI}(M_{i},\tilde{M}_{ij})$ in $3$ representative cases of (1) $0 \le\delta_{i}\le Q3$ (2) $\delta_{i} > Q3$ (3) $\delta_{i} < 0$. }
  \label{fig:3 cases}
\end{figure}

To quantify the robustness of barycenter, we record the Frobenius distance between the barycenters of matrices with and without perturbations, i.e., $d_{F}(M_{i}^{\circ},\tilde{M}_{ij}^{\circ})$, for AI and BW measure, respectively. The smaller the distance $d_{F}(M_{i}^{\circ},\tilde{M}_{ij}^{\circ})$ is, the more robust the barycenter is.
Note that the Frobenius distance between two matrices is defined as
$$d_{F}(X,Y) = ||X - Y||_F=\sqrt{tr[(X-Y)^{*}(X-Y)]}.$$
For each pair of $(A_i, B_i)$, the distribution of $\{d_{F}(M_{i}^{\circ},\tilde{M}_{ij}^{\circ}), j=1, \ldots, 100\}$ for AI and BW distance are compared. Figure \ref{fig:3 cases} shows three representative cases of $\{d_{F}(M_{i}^{\circ},\tilde{M}_{ij}^{\circ})\}$.
To quantify the differences between the two distributions, i.e., $\{d_{F}(M_{i}^{AI},\tilde{M}_{ij}^{AI})\}$ and $\{d_{F}(M_{i}^{BW},\tilde{M}_{ij}^{BW})\}$, for each $(A_i, B_i)$ pair, we record the relative difference of the two sample means as the evaluation metric:
$$ \delta_{i} = \frac{\overline{d_{F}(M_{i}^{AI},\tilde{M}_{ij}^{AI})} - \overline{d_{F}(M_{i}^{BW},\tilde{M}_{ij}^{BW})}}{\overline{d_{F}(M_{i}^{BW},\tilde{M}_{ij}^{BW})}} $$
$$\overline{d_{F}(M_{i}^{\circ},\tilde{M}_{ij}^{\circ})} = \frac{1}{100}\sum_{j = 1}^{100}d_{F}(M_{i}^{\circ},\tilde{M}_{ij}^{\circ})$$
where $i=1, \ldots, 200$. $\delta_i$ quantifies the differences in the robustness of BW and AI barycenter. If $\delta_i$ is positive (e.g., case 1 and 2 in Figure~\ref{fig:3 cases}), the BW barycenter is more robust.

Table \ref{tab:5 statistics} summarizes the distribution of $\{\delta_i\}$ for $100$ pairs. In Figure \ref{fig:relative difference}, case 1 represents the most common scenario ($73\%$) where $d_{F}(M_{i}^{BW},\tilde{M}_{ij}^{BW})$ is mostly smaller than $d_{F}(M_{i}^{AI},\tilde{M}_{ij}^{AI})$. Case 2 shows the common scenario ($25\%$) when $d_{F}(M_{i}^{BW},\tilde{M}_{ij}^{BW})$ is significantly smaller than $d_{F}(M_{i}^{AI},\tilde{M}_{ij}^{AI})$. Both case 1 and 2 represents cases when BW barycenter is more robust than AI barycenter. Case 3 depicts the least common scenario ($2\%$) where AI barycenter shows greater robustness than BW barycenter ($\delta_{i}<0$), but with a considerable amount of overlap between the two distributions. Note that the maximum and minimum of $\{\delta_{i}\}$ are $35.04$ and $-0.423$ respectively, which implies that AI barycenter is not significant more robust than BW barycenter even in the rare occurrence of case 3.

\begin{table}[htbp]
\normalsize
    \caption{Descriptive statistics of the distribution of $\{\delta_{i}, i=1, \ldots, 100\}$ with $n = 10$ and $p = 0.2$. }
    \label{tab:5 statistics}
    \centering
    \begin{tabular}{p{1cm}p{1cm}p{1cm}p{1cm}p{1cm}p{1cm}}
        \hline
        Min & Q1 & Median & Mean & Q3 & Max \\
        \hline
        -0.42 & 2.67 & 8.23 & 9.53 & 14.81 & 35.04 \\
        \hline
    \end{tabular}
\end{table}

With the proportion of close-to-zero eigenvalues $p$ being $0.2$, we investigate how matrix dimension $n$ influence the robustness of BW and AI barycenters. Figure \ref{fig:relative difference} (A) shows the empirical distribution of $\{\delta\}$ with varying $n$. The line plot in the top right corner describes the median of each distribution of $\{\delta\}$. Overall, the majority of $\delta$ locates on the positive sides regardless of $n$ while negative $\delta$ rarely occur in our simulations. It indicates that BW barycenter is more robust than AI barycenter, i.e., the BW barycenter is less sensitive to changes in the matrices. Besides, as $n$ increases, the distribution of $\delta_{i}$ becomes more concentrated towards $0$ and a monotonic decrease in the median is also observed. In another word, the robustness superiority of BW barycenter over AI is more significant for lower dimensional matrices (p-value $< 0.05$ for Kruskal-Wallis test and $11$ out of the $15$ pairwise comparisons of Dunn test are significant).

\begin{figure}[htbp]
\centering
  \includegraphics[width=0.48\textwidth]{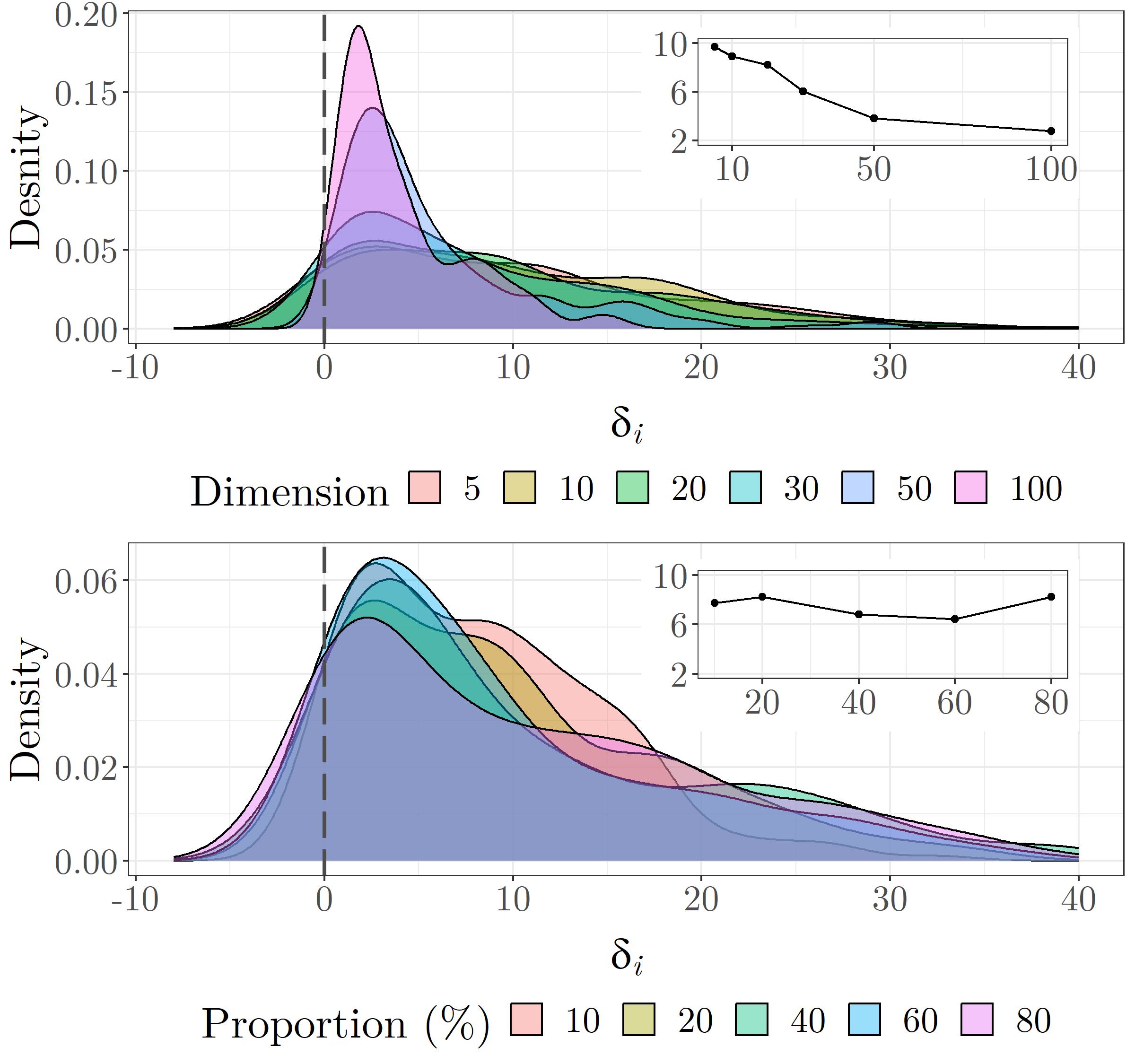}
  \caption{Empirical distributions of $\delta_{i}$ with (A) varying matrix dimension (B) varying proportion of close-to-zero eigenvalues. }
  \label{fig:relative difference}
\end{figure}

With matrix dimension $n$ being $10$, we also study the impact of $p$ on the robustness of the two barycenters, as shown in Figure \ref{fig:relative difference} (B). Overall, majority of $\delta$ lies in the positive side, which implies that BW barycenter is more robust than AI barycenter. Different from $n$, the distribution of $\delta$ with varying $p$ share similar shapes with similar medians (p-value is not significant for Kruskal-Wallis test). Therefore, $p$ has little impact on the differences in the robustness of BW and AI barycenter.

\subsection{Accuracy, Efficiency, and Robustness of Barycenter Estimation for More Than Two Matrices}

Now, we consider a set of positive definite matrices, $A_1, \ldots, A_m, m>2$ with $p$ being 0.2. The number of matrices $m$ and matrix dimension $n$ are chosen according to the grid given in Table~\ref{tab:simulation setting}. In this simulation, we comprehensively study the accuracy, efficiency, and robustness of the proposed three algorithms for barycenter estimation using BW distance. For stopping criteria, we use $\varepsilon = 10^{-3}$.

\subsubsection{Accuracy of Barycenter Estimation with BW distance}

\begin{figure}[!tbp]
\centering
  \includegraphics[width=0.49\textwidth]{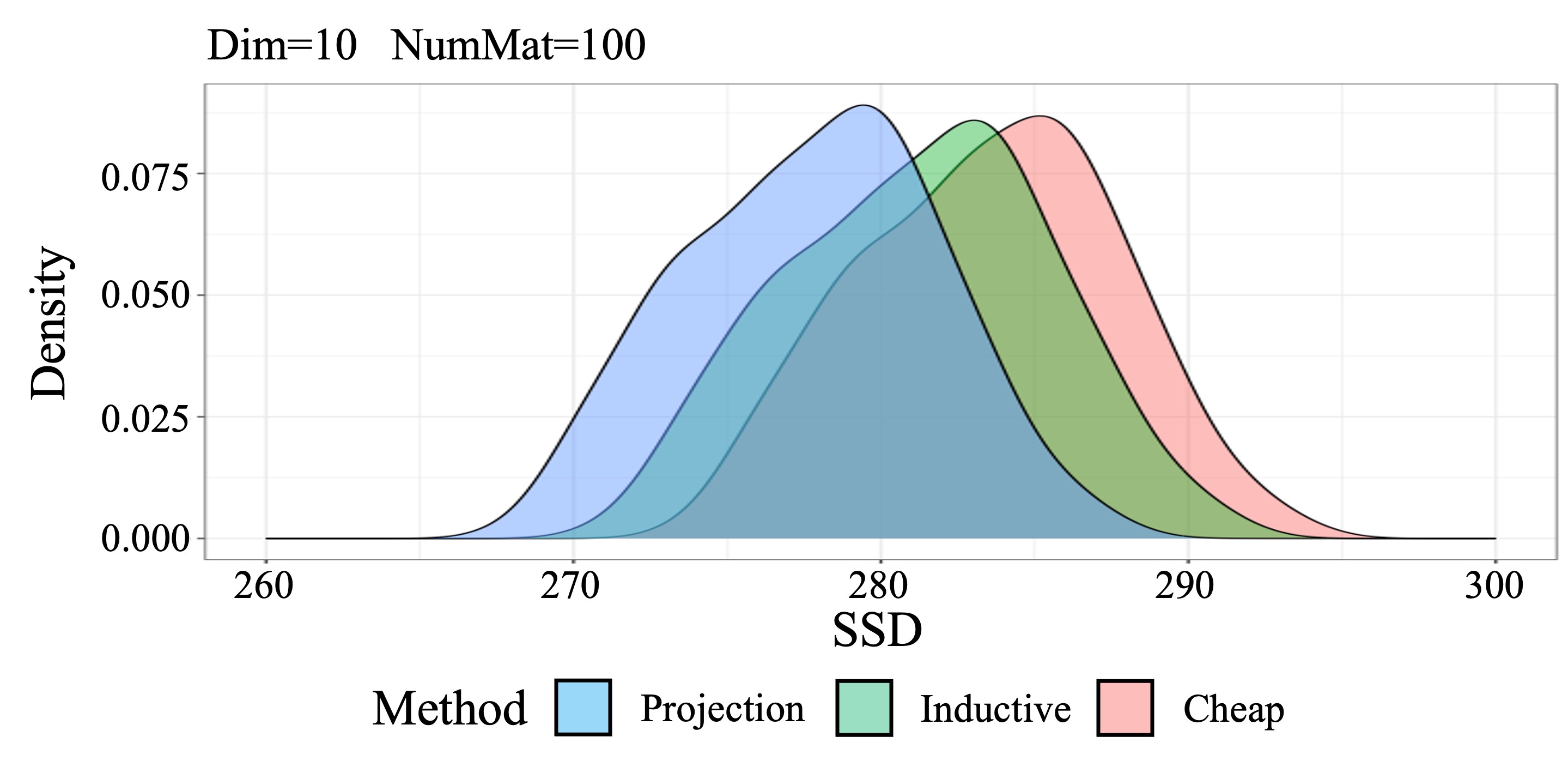}
  \caption{Empirical distribution of $\{SSD^{\circ}\}$ for the three mean algorithms, with matrix dimension being 10 and number of matrix being 100.}
  \label{fig:accuracy SSD}
\end{figure}

Though the true Fr\'echet mean is unknown, it is the one that minimizes the sum of the squared BW distances between each matrix and itself by definition. Therefore, the sum of the squared distances (SSD) is chosen to be the evaluation metric for the accuracy of the barycenter estimation. The smaller the SSD is, more accurate the barycenter estimation is.
$$ SSD^{\circ} = \sum_{i=1}^{m} d^{2}_{BW}(A_{i},M^{\circ})$$
where $M^{\circ}$ denotes the barycenter estimated by the three algorithms. I denotes Inductive mean algorithm, P denotes Projection mean algorithm, and C denotes Cheap Mean algorithms in the following.

Take $n=10$ and $m=100$ as an example. In the $k^{th}$ iteration ($k=1, \ldots, 100$), $100$ positive definite matrices are randomly generated with dimensions being $10$ and $2$ eigenvalues being close-to-zeros (e.g., $O(10^{-3})$. The barycenter of the $100$ matrices (i.e., $M^I, M^P, M^C$) are then obtained via the three algorithms respectively. The SSD of each barycenter is obtained accordingly, denote them as $SSD^I_k, SSD^P_k, SSD^C_k$. The process is iterated $100$ times, resulting in a collection of $\{SSD^{\circ}_k, k=1, \ldots, 100\}$. Figure \ref{fig:accuracy SSD} shows the empirical distribution of $\{SSD^{I}\}, \{SSD^{P}\}, \{SSD^{C}\}$. $\{SSD^{P}\}$ is significantly less than the other two with p-values being $<0.05$ in Wilcoxon signed-rank tests. Therefore, projection mean algorithm is the most accurate among the three methods.

\begin{figure*}[h]
\centering
  \includegraphics[width=\textwidth]{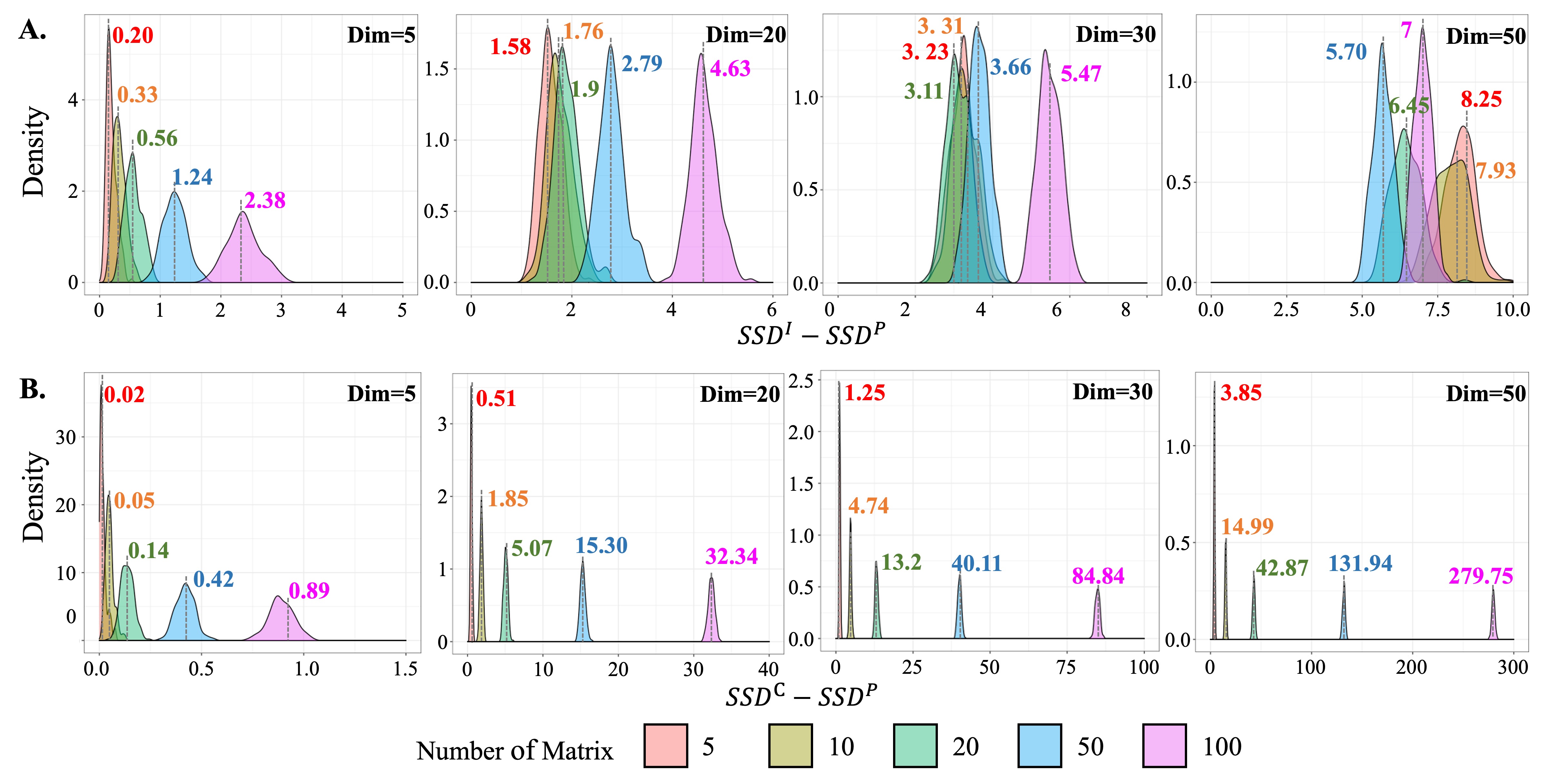}
  \caption{Accuracy of barycenter estimation algorithms. (A) Distribution of $(SSD^{I}-SSD^{P})$. (B) Distribution of $(SSD^{C}-SSD^{P})$. All distributions are significantly different from each other (all p-values$<0.05$).}
  \label{fig:accuracy all}
\end{figure*}

We then change the matrix dimension $n$ and the number of matrix $m$, and perform the aforementioned process. For the comparison purpose, we denote $\{SSD^I_k-SSD^P_k\}_{k=1}^{100}$ as $SSD^I-SSD^P$ and $\{SSD^C_k-SSD^P_k\}_{k=1}^{100}$ as $SSD^C-SSD^P$. Figure \ref{fig:accuracy all} shows the empirical distributions of $SSD^I-SSD^P$ and $SSD^C-SSD^P$.
Overall, both $SSD^I-SSD^P$ and $SSD^C-SSD^P$ are positive, which implies that the project mean algorithm is the most accurate, regardless of $n$ and $m$.
Compared with $SSD^I-SSD^P$, $SSD^C-SSD^P$ is in a much larger scale except for dimension being 5. It indicates that inductive mean algorithm is generally more accurate than cheap mean algorithm, except for low dimensional matrices.
Besides, cheap mean algorithm is less accurate when the number of matrices or dimension of matrices increases (Figure \ref{fig:accuracy all} (B)).

The differences between inductive mean and projection mean also increases when the matrix dimension is higher. However, unlike cheap mean, the influence of matrix number $m$ is not monotonic for inductive mean algorithm. For instance, when $n$ is 50, the smallest SSD happens when $m$ is also 50. When $n$ is 30, the smallest SSD occurs when $m$ is 20. Therefore, we suspect that inductive mean algorithm is most accurate when the number of matrices is close to the dimension of matrices.

In summary, the barycenter estimated by projection mean algorithm is the most accurate regardless of the number of matrices and dimension of matrices. Inductive mean algorithm is more accurate than cheap mean algorithm except for low dimensional matrices.

\subsubsection{Efficiency of Barycenter Estimation with BW distance}

To evaluate the efficiency of barycenter estimation algorithms, the running time is chosen as the evaluation metric. The efficiency of the proposed three algorithms with BW distance is investigated and compared with the algorithms coupled with AI distance. Besides, the contributing factors that affect the efficiency of each algorithm is also comprehensively studied.

\begin{figure}[!tbp]
\centering
  \includegraphics[width=0.49\textwidth]{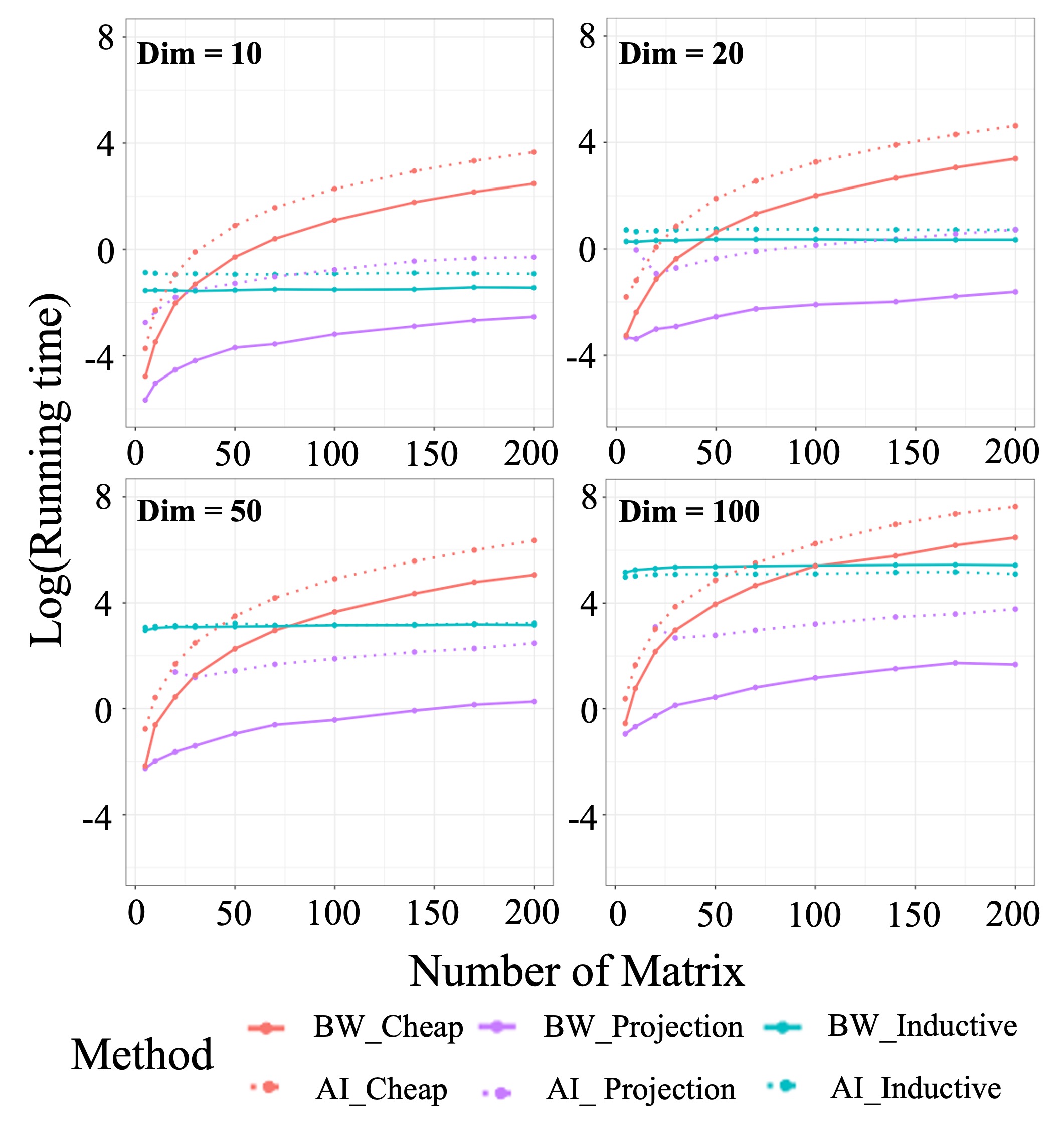}
  \caption{Comparison of efficiency of Barycenter estimation algorithms coupled with BW and AI distance.}
  \label{fig:efficiency more than two}
\end{figure}

Similarly, positive definite matrices are randomly generated with different $n$ and $m$ chosen from Table \ref{tab:simulation setting}. The running time is averaged over 100 iterations. The details of the barycenter algorithms coupled with AI distance are discussed in \cite{jeuris2012survey}.
Figure \ref{fig:efficiency more than two} shows the comparison of the running time (seconds) in log scale.

Comparing the two distances, both projection mean and cheap mean algorithms performs more efficiently when coupled with BW distance regardless of matrix dimension ($n$) and number of matrices ($m$). For inductive mean algorithm, it is more efficient when coupled with BW distance for low dimensional matrices.
The impact of the contributing factors (i.e., $n$ and $m$) is similar for both distances. For instance, regardless of the distance choice, the efficiency of inductive mean algorithm is mainly affected by matrix dimension; the cheap mean and projection mean algorithms become slower when $n$ or $m$ is larger.

For BW distance, the projection mean algorithm is the most efficient regardless of $n$ and $m$. Depending on $n$, cheap mean algorithm is more efficient when $m$ is small and less efficient when $m$ becomes large compared with inductive mean algorithm.
For AI distance, the cheap mean algorithm is the most efficient when $m$ is small, for instance, less than 10 matrices. This is consistent with the finding in \cite{jeuris2012survey}. For more matrices (i.e., $m$ is large), projection mean and inductive mean become more efficient. Depending on the matrix dimension, projection mean algorithm is more efficient when $n$ is large. However, it is worth noting that projection mean algorithm does not converge when it is coupled with AI distance and $n$ is much larger than $m$ \cite{jeuris2012survey}.

Overall, the projection mean algorithm coupled with BW distance is the most efficient regardless of $n$ and $m$.

\subsubsection{Robustness of Barycenter Estimation with BW distance}

The robustness of barycenter estimation refers to the changes in the estimated barycenter caused by perturbations that contaminate the matrices. Therefore, we use the Frobenius distance between the estimated Fr\'echet mean with and without perturbations as the evaluation metric to quantify the robustness of algorithms.

With $n$ and $m$ chosen from Table \ref{tab:simulation setting}, in the $i^{th}$ iteration $(i=1, \ldots, 100)$, positive definite matrices $\{A_{i}^{1},...,A_{i}^{m}\}$ are randomly generated.
The contaminated matrices, $\{\tilde{A}_{ij}^{1},...,\tilde{A}_{ij}^{m},j=1, \ldots,100\}$, are then obtained by adding the randomly generated Hermitian perturbation matrices $E_{ij}^{k}$ on $A^k_i$ as follows:
$$\tilde{A}_{ij}^{k}=A_{i}^{k}+E_{ij}^{k},i,j=1,...,100,k=1,...,m.$$
The robustness of barycenter estimation is then quantified as the Frobenius distance between the estimated Fr\'echet mean with and without perturbation:
$$d_{F,ij}^{\circ} = ||M_{i}^{\circ},\tilde{M}_{ij}^{\circ}||_{F},i,j=1,...,100$$
where $M_{i}^{\circ}$ and $\tilde{M}_{ij}^{\circ}$ denotes the barycenter of matrices estimated by difference algorithms without and with perturbations respectively. For the comparison of robustness, we focus on three algorithms, inductive mean with AI distance, inductive mean with BW distance, and projection mean with BW distance because cheap mean is time consuming especially for large $m$ and projection mean with AI does not converge for small $m$.

\begin{figure}[!tbp]
\centering
  \includegraphics[width=0.49\textwidth]{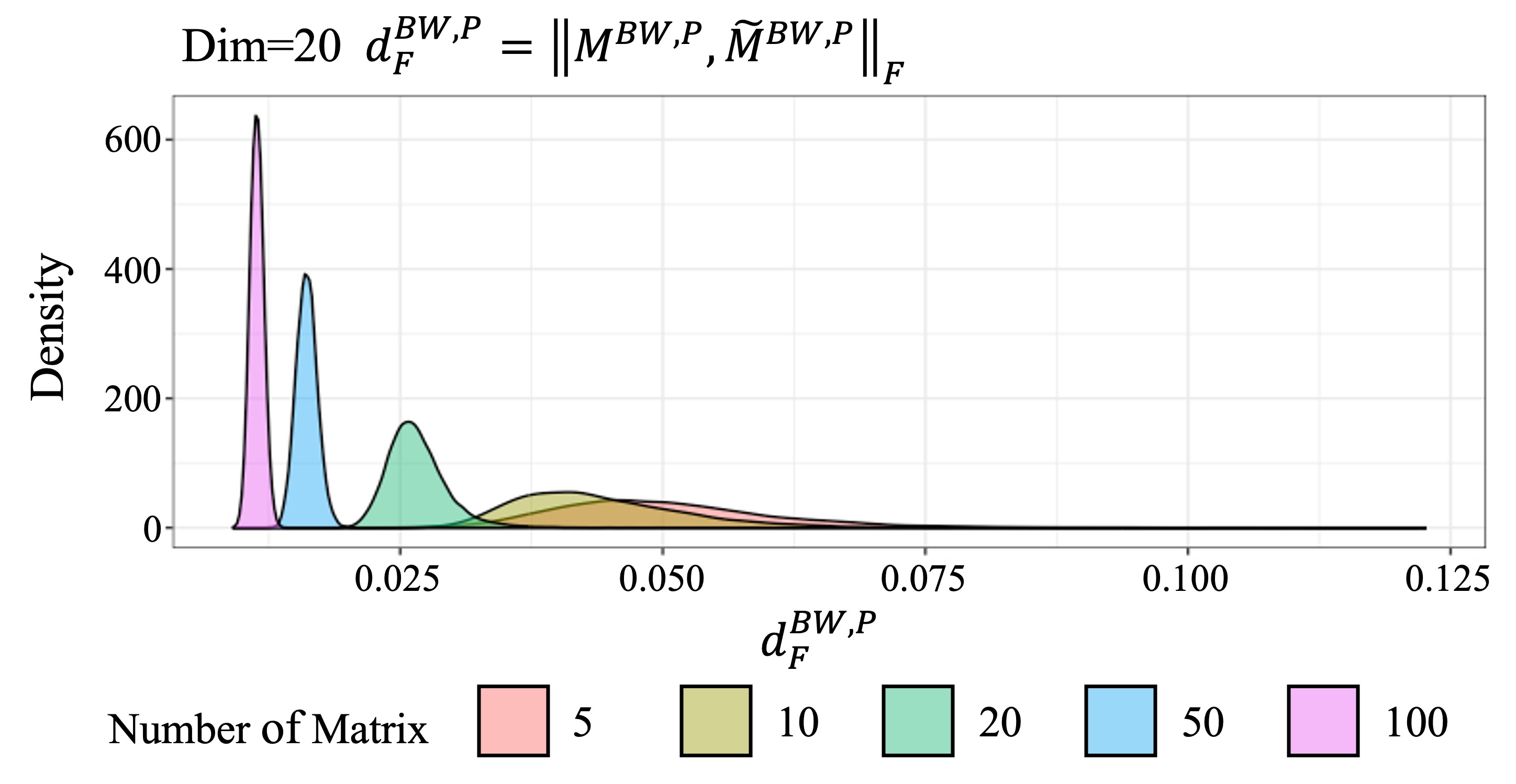}
  \caption{Distribution of Frobenious distance between the Fr\'echet mean estimated by BW projection mean algorithm for $20\times20$ matrices without and with perturbations.}
  \label{fig:robustness_bw_projection}
\end{figure}

With $n$ being 20, Figure \ref{fig:robustness_bw_projection} displays the distribution of $\{d_{F,ij}^{BW,P}\}$. The contributing factor $m$ has significant impact on the robustness of projection mean algorithm (BW). The estimated barycenter is more robust when there are more matrices (all p-values $< 0.05$ for Kruskal-Wallis test and Pairwise Wilcoxon Rank Sum test).

\begin{figure*}[h]
\centering
  \includegraphics[width=\textwidth]{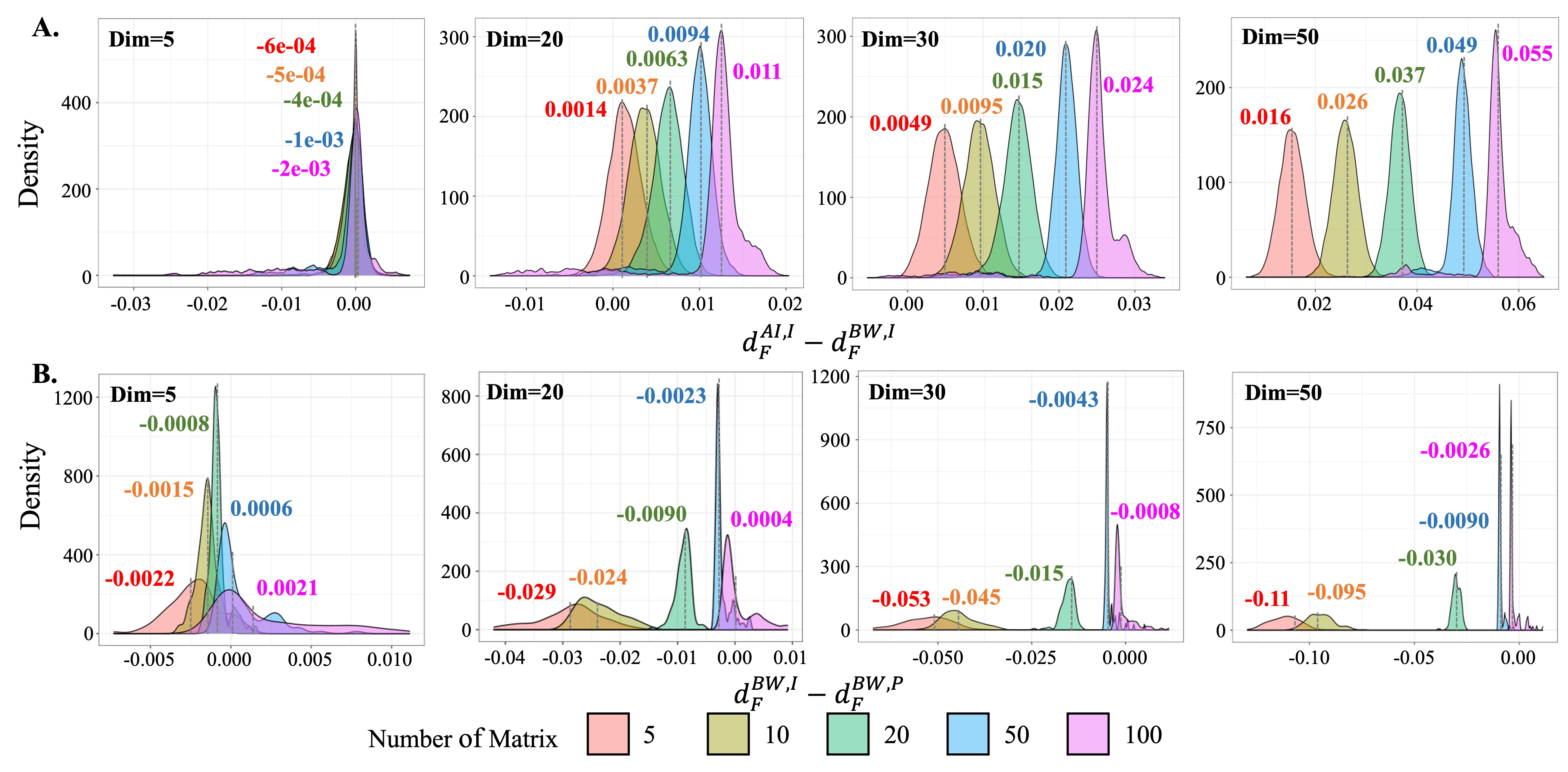}
  \caption{Comparison of robustness of Barycenter estimation for more than two matrices. (A) Distribution of $d_{F}^{AI,I}-d_{F}^{BW,I}$ with different $n$ and $m$. (B) Distribution of $d_{F}^{BW,I}-d_{F}^{BW,P}$ with different $n$ and $m$.}
  \label{fig:robustness_more than two}
\end{figure*}

To compare the robustness between two algorithms,the following metrics are displayed in Figure \ref{fig:robustness_more than two}.
$$ d_F^{AI,I}-d_F^{BW,I} := \{d_{F,ij}^{AI,I} - d_{F,ij}^{BW,I},i,j=1,...,100\} $$
$$ d_F^{BW,I}-d_F^{BW,P} := \{d_{F,ij}^{BW,I} - d_{F,ij}^{BW,P},i,j=1,...,100\} $$


Figure \ref{fig:robustness_more than two} (A) displays the differences in the robustness of inductive mean algorithms coupled with AI and BW distances. As $n$ or $m$ increases, BW inductive mean algorithm becomes significantly more robust than AI mean (all p-values $< 0.05$ for Kruskal-Wallis test and Pairwise Wilcoxon Rank Sum test).

Figure \ref{fig:robustness_more than two} (B) compares the BW inductive and BW projection mean algorithms. The BW projection mean algorithm is more robust than BW inductive algorithm when $m$ is much larger than $n$. For instance, when there are less than 20 $5\times5$ matrices, BW inductive is more robust but when there are 50 or 100 matrices, BW projection algorithm becomes significantly more robust (p-values $< 0.05$ for both Kruskal-Wallis test and Pairwise Wilcoxon Rank Sum test).

In summary, BW projection mean algorithm is the most robust when $m$ is much larger than $n$. BW inductive mean algorithm is also robust compared to AI inductive mean algorithm.

\section{Conclusion} \label{conclusion}
In this paper, we first establish the mathematical foundation for the BW distance by studying the properties of BW distance and the retraction maps of the manifold $\overline{\P}_n$. To characterize the distribution of PSD matrices on the manifold, we propose three algorithms to estimate the Fr\'echet mean (i.e., barycenter) of a set of PSD matrices. With extensive simulation experiments, we comprehensively investigate three aspects: 1. the robustness of BW distance when using it to quantify the distance between PSD matrices, 2. the robustness of BW barycenter for two matrices if the two matrices are contaminated by small perturbations, 3. the accuracy, efficiency, and robustness of the proposed three barycenter estimation algorithms. Compared with AI distance, BW distance is more robust especially when matrices are close to being positive semi-definite, which is a common scenario for high dimensional data. When there are only two matrices, BW barycenter is more robust than AI barycenter when matrices are affected by noises, especially for positive definite matrices with some small eigenvalues. When there are more than two matrices, BW projection mean algorithm outperforms others in terms of accuracy, efficiency, and robustness. Therefore, BW distance and projection mean algorithm are recommended especially for studying high dimensional matrices.


%

\appendices


\section*{Acknowledgment}

The authors would like to thank anonymous referees, an Associate Editor, and the Editor for their constructive comments that improved the quality of this paper.
This paper is based upon work supported by the National Science Foundation under Grant No. 2153492.

\ifCLASSOPTIONcaptionsoff
  \newpage
\fi



%

\bibliography{references}
\bibliographystyle{IEEEtran}

\end{document}